\documentclass[11pt]{article}
\usepackage[letterpaper,margin=1in]{geometry}
\usepackage{amsmath,amsthm,amssymb}
\usepackage{algorithm}
\usepackage{algpseudocode}
\usepackage[utf8]{inputenc}
\usepackage[T1]{fontenc}
\usepackage{lmodern}
\usepackage{xspace}
\usepackage[hidelinks,
						bookmarksnumbered,
						bookmarksopen=true
						]{hyperref}
\usepackage[noabbrev,capitalise,nameinlink]{cleveref}
\usepackage{enumitem}

\usepackage{dsfont}

\theoremstyle{plain}
\newtheorem{theorem}{Theorem}[section]
\newtheorem{lemma}[theorem]{Lemma}

\newtheorem{proposition}[theorem]{Proposition}

\newtheorem{corollary}[theorem]{Corollary}

\theoremstyle{definition}

\newcommand*\ie{i.\kern.1em e.\ }
\newcommand*\eg{e.\kern.1em g.\ }

\newcommand{\dist}{\mathsf{dist}}
\newcommand{\dgn}{\mathsf{\delta}}

\renewcommand{\Pr}[1]{\bP \left[ #1 \right]} 
\newcommand{\Pru}[2]{\underset{ #1 }\bP \left[ #2 \right]}
\newcommand{\Pruc}[3]{\underset{ #1 }\bP \left[ #2 \;\; \mid \;\; #3 \right]}

\newcommand{\zo}{\{0,1\}}

\newcommand{\cF}{\ensuremath{\mathcal{F}}}
\newcommand{\cG}{\ensuremath{\mathcal{G}}}
\newcommand{\cH}{\ensuremath{\mathcal{H}}}

\newcommand{\bN}{\ensuremath{\mathbb{N}}}
\newcommand{\bP}{\ensuremath{\mathbb{P}}}

\newcommand\splitaftercomma[1]{%
  \begingroup
  \begingroup\lccode`~=`, \lowercase{\endgroup
    \edef~{\mathchar\the\mathcode`, \penalty0 \noexpand\hspace{0pt plus .25em}}%
  }\mathcode`,="8000 #1%
  \endgroup
}

\title{Optimal Adjacency Labels for Subgraphs of Cartesian Products\footnote{The results of this paper previously appeared in the proceedings of the 50th International Colloquium on Automata, Languages and Programming, ICALP 2023 \cite{EHZ23}.}}

\author{Louis Esperet\thanks{Partially supported by the French ANR Projects
  GATO (ANR-16-CE40-0009-01), GrR (ANR-18-CE40-0032), TWIN-WIDTH
  (ANR-21-CE48-0014-01) and by LabEx
  PERSYVAL-lab (ANR-11-LABX-0025).}\\ Laboratoire G-SCOP, CNRS, France \\ \texttt{louis.esperet@grenoble-inp.fr}
\and
Nathaniel Harms\thanks{This work was partly funded by NSERC, while the author was
a student at the University of Waterloo, visiting Laboratoire G-SCOP and the University of Liverpool.} \\ EPFL, Switzerland \\
\texttt{nathaniel.harms@epfl.ch}
\and Viktor Zamaraev  \\ University of Liverpool, UK \\ \texttt{viktor.zamaraev@liverpool.ac.uk}}

\begin{document}

\maketitle

\begin{abstract}
For any hereditary graph class $\cF$, we construct optimal adjacency labeling schemes for the
classes of subgraphs and induced subgraphs of Cartesian products of graphs in $\cF$. As a
consequence, we show that, if $\cF$ admits efficient adjacency labels (or, equivalently, small
induced-universal graphs) meeting the information-theoretic minimum, then the classes of subgraphs
and induced subgraphs of Cartesian products of graphs in $\cF$ do too. Our proof uses ideas
from randomized communication complexity, hashing, and additive combinatorics, and improves upon recent results of Chepoi,
Labourel, and Ratel [Journal of Graph Theory, 2020].
\end{abstract}

\thispagestyle{empty}
\setcounter{page}{0}
\newpage

\section{Introduction}

In this paper, we present optimal adjacency labelling schemes (equivalently, induced-universal graph
constructions) for subgraphs of Cartesian products, which essentially closes a recent line of work
studying these objects \cite{CLR20, Har20, AAL21, HWZ21, AAA+22, EHK22}. To do so, we introduce a
few new techniques for designing adjacency labelling schemes.

\paragraph*{Adjacency labeling.}
A \emph{class} of graphs is a set $\cF$ of graphs closed under isomorphism, where the set $\cF_n
\subseteq \cF$ of graphs on $n$ vertices has vertex set $[n]$. It is \emph{hereditary} if it is also
closed under taking induced subgraphs, and \emph{monotone} if it is also closed under taking
subgraphs. An \emph{adjacency labeling scheme} for a class $\cF$ consists of a \emph{decoder} $D :
\zo^* \times \zo^* \to \zo$ such that for every $G \in \cF$ there exists a \emph{labeling} $\ell :
V(G) \to \zo^*$ satisfying
\[
  \forall x,y \in V(G) \,:\qquad D(\ell(x), \ell(y)) = 1 \iff xy \in E(G) \,.
\]
The \emph{size} of the adjacency labeling scheme (or \emph{labeling scheme} for short) is the
function $n \mapsto \max_{G \in \cF_n} \max_{x \in V(G)} |\ell(x)|$, where $|\ell(x)|$ is the number
of bits of $\ell(x)$. Labeling schemes have been studied extensively since their introduction by
Kannan, Naor, \& Rudich \cite{KNR92} and Muller \cite{Mul89}. If $\cF$ admits a labeling scheme of
size $s(n)$, then a graph $G \in \cF_n$ can be recovered from the $n \cdot s(n)$ total bits in the
adjacency labels of its vertices, so a labeling scheme is an encoding of the graph, distributed
among its vertices. The information-theoretic lower bound on any encoding is $\log |\cF_n|$, so the
question is, when can the distributed adjacency labeling scheme approach this bound? In other words,
which classes of graphs admit labeling schemes of size $O(\tfrac{1}{n} \log |\cF_n|)$? We will say
that a graph class has an \emph{efficient} labeling scheme if it either has a labeling scheme of
size $O(1)$ (\ie it satisfies $\log |\cF_n| = o(n \log n)$ \cite{Sch99}), or $O(\tfrac{1}{n} \log
|\cF_n|)$.

\newcommand{\her}{\mathsf{her}}
\newcommand{\mon}{\mathsf{mon}}
\paragraph*{Cartesian products.}
Write $G \square H$ for the Cartesian product of $G$ and $H$, write $G^d$ for the $d$-wise Cartesian
product of $G$, and for any class $\cF$ write $\cF^\square = \{ G_1 \square G_2 \square \dotsm
\square G_d : d \in \bN, G_i \in \cF \}$ for the class of Cartesian products of graphs in $\cF$. A
vertex $x$ of $G_1 \square \dotsm \square G_d$ can be written $x = (x_1,\dotsc,x_d)$ where $x_i \in
V(G_i)$ and two vertices $x,y$ are adjacent if and only if they differ on exactly one coordinate $i
\in [d]$, and on this coordinate $x_iy_i \in E(G_i)$. Write $\mon(\cF^\square)$ and
$\her(\cF^\square)$, respectively, for the monotone and hereditary closures of this class, which are
the sets of all graphs $G$ that are a subgraph (respectively, induced subgraph) of some $H \in
\cF^\square$. 

The main result of this paper is to construct optimal labeling schemes for $\mon(\cF^\square)$ and
$\her(\cF^\square)$ from an optimal labeling scheme for $\cF$.  Cartesian products appear several
times independently in the recent literature on labeling schemes \cite{CLR20,Har20,AAL21} (and later
in \cite{HWZ21,AAA+22,EHK22}), and are extremely natural for the problem of adjacency labeling for a
few reasons.

First, for example, if $\cF$ is the class of complete graphs, a labeling scheme for
$\her(\cF^\square)$ is equivalent to an encoding $\ell : T \to \zo^*$ of strings $T \subseteq
\Sigma^*$, with $\Sigma$ being an arbitrarily large finite alphabet, such that a decoder who doesn't
know $T$ can decide whether $x,y \in T$ have Hamming distance 1, using only the encodings $\ell(x)$
and $\ell(y)$.  Replacing complete graphs with, say, paths, one obtains induced subgraphs of grids
in arbitrary dimension.  Switching to $\mon(\cF^\square)$ allows arbitrary edges of these products
to be deleted.

Second, Cartesian product graphs admit by definition a natural but inefficient ``implicit
representation''\!\!, meaning (informally) that the adjacency between two vertices $x$ and $y$ can
be verified by examining their representation (in this case, the tuples $x = (x_1,\dotsc,x_d)$ and
$y = (y_1, \dotsc, y_d)$).  Formalizing and quantifying this general notion was the motivation for
labeling schemes \cite{KNR92}. \cite{KNR92} observed that adjacency labeling schemes are equivalent
to \emph{induced-universal graphs} (or simply \emph{universal graphs}). A sequence of graphs
$(U_n)_{n \in \bN}$ are universal graphs of size $n \mapsto |U_n|$ for a class $\cF$ if each
$n$-vertex graph $G \in \cF$ is an induced subgraph of $U_n$. A labeling scheme of size $s(n)$ is
equivalent to a universal graph of size $2^{s(n)}$. If $(U_n)_{n \in \bN}$ are universal graphs for $\cF$ then for
large enough $d = d(n)$, the graphs $(U_n^d)_{n \in \bN}$ are universal for $\her(\cF^\square)$, but in general
this construction has exponential size: the hypercubes $K_2^d$ are themselves universal for
$\her(\{K_2\}^\square)$, but a star with $n-1$ leaves cannot be embedded in $K_2^d$ for $d < n-1$,
so these universal graphs are of size at least $2^{n-1}$. It is not clear \emph{a priori} whether it
is possible to use the universal graphs for the base class $\cF$ to obtain more efficient universal
graphs for $\her(\cF^\square)$, and even less clear for $\mon(\cF^\square)$.

Finally, there was the possibility that \emph{subgraphs} of Cartesian products could provide the
first explicit counterexample to the Implicit Graph Conjecture (IGC) of \cite{KNR92,Spin03}, which
suggested that the condition $\log |\cF_n| = O(n \log n)$ was \emph{sufficient} for $\cF$ to admit a
labeling scheme of size $O(\log n)$; this was refuted by a non-constructive argument in a recent
breakthrough of Hatami \& Hatami \cite{HH21}.  There is a labeling scheme of size $O(\log^2 n)$ for
the subgraphs of hypercubes, due to a folklore bound of $\log n$ on the degeneracy of this class
(see \cite{Gra70}) and a general $O(\dgn \log n)$ labeling scheme for classes of degeneracy $\dgn$
\cite{KNR92}.  Designing an efficient labeling scheme for \emph{induced} subgraphs of hypercubes
(rather, the weaker question of proving bounds on $|\cF_n|$ for this family) was an open problem of
Alecu, Atminas, \& Lozin \cite{AAL21}, resolved concurrently and independently in \cite{Har20} using
a probabilistic argument; this also gave an example of a class with an efficient labeling scheme but
unbounded \emph{functionality}, answering another open question of \cite{AAL21}.  Also
independently, Chepoi, Labourel, \& Ratel \cite{CLR20} studied the structure of general Cartesian
products, motivated by the problem of designing labeling schemes for the classes
$\mon(\cF^\square)$. They give upper bounds (via bounds on the degeneracy) for a number of special
cases but do not improve on the $O(\log^2 n)$ bound for hypercubes.

It is shown in \cite{EHK22} that, while \emph{induced} subgraphs of hypercubes have a constant-size
\emph{adjacency sketch} (a probabilistic version of a labeling scheme), the \emph{subgraphs} of
hypercubes do not. This gave a natural counterexample to a conjecture of \cite{HWZ21}, whose earlier
refutation by a more specialized construction \cite{HHH21a} led to the refutation of the IGC
\cite{HH21}. Since (1) subgraphs and induced subgraphs of hypercubes are significantly different in
the sketching model, (2) the earlier refutation of the \cite{HWZ21} conjecture led to a refutation
of the IGC, and (3) the previous  work considering Cartesian products \cite{CLR20, Har20, HWZ21,
AAL21, AAA+22} had not improved on the trivial $O(\log^2 n)$ bound for subgraphs, \cite{EHK22} asked
whether subgraphs of Cartesian products could give the first explicit counterexample to the IGC.
Alas, we find that this is not so.

\paragraph*{Results and techniques.}
We improve the best-known $O(\log^2 n)$ bound for subgraphs of hypercubes to the optimal $O(\log
n)$, and in general show how to construct optimal labels for all subgraphs and induced subgraphs of
Cartesian products. Our proof is short, and departs significantly from standard techniques in the
field of labeling schemes: we do not rely on any structural results, graph width parameters, or
decompositions, and instead use communication complexity (as in \cite{Har20,HWZ21}), encoding,
hashing arguments, and a construction from additive combinatorics, all of which may be useful for future work on labeling schemes.  We prove:

\begin{theorem}
\label{thm:main}
Let $\cF$ be a hereditary class with an adjacency labeling scheme of size $s(n)$. Then:
\begin{enumerate}[itemsep=0pt]
\item $\her(\cF^\square)$ has a labeling scheme of size at most $2s(n) + O(\log n)$.
\item $\mon(\cF^\square)$ has a labeling scheme of size at most $2s(n) + O(\dgn(n) + \log n)$, where $\dgn(n)$ is the maximum degeneracy of any $n$-vertex graph in $\mon(\cF^\square)$.
\end{enumerate}
\end{theorem}

\noindent
We allow $\cF$ to be finite, in which case $s(n) = O(1)$; in particular, setting $\cF = \{ K_2, K_1
\}$, we get the result for hypercubes:

\begin{corollary}
Let $\cH$ be the class of hypercube graphs. Then $\mon(\cH)$ has a labeling scheme of size $O(\log
n)$.
\end{corollary}

\noindent
All of the labeling schemes of Chepoi, Labourel, \& Ratel \cite{CLR20} are obtained by bounding the degeneracy $\dgn(G)$ of a graph $G$
and applying as a black-box the labeling scheme of size $O(\dgn(G) \cdot \log n)$ from \cite{KNR92}. For example, they
get labels of size $O(d \log^2 n)$ when the base class $\cF$ has degeneracy $d$, by showing that
$\mon(\cF^\square)$ has degeneracy $O(d \log n)$. Our result can be substituted for that black-box,
replacing the multiplicative $O(\log n)$ with an \emph{additive} $O(\log n)$, thereby improving all
of the results of \cite{CLR20} when combined with their bounds on $\dgn(G)$; for example, achieving
$O(d \log n)$ when $\cF$ has degeneracy $d$.

For subgraphs of hypercubes, \cite{CLR20} observed that a bound of $O(\mathsf{vc}(G) \log n)$
follows from the inequality $\dgn(G) \leq \mathsf{vc}(G)$ due to Haussler \cite{Hau95}, where
$\mathsf{vc}(G)$ is the VC dimension\footnote{See \cite{CLR20} for the definition of VC dimension},
which can be as large as $\log n$ but is often much smaller; they generalize this inequality in
various ways to other Cartesian products.  Our result supercedes the VC dimension result for
hypercubes.

\cref{thm:main} is optimal up to constant factors (which we have not tried to optimize), and yields
the following corollary (see \cref{section:optimality} for proofs).

\begin{corollary}
\label{cor:efficient}
If a hereditary class $\cF$ has an efficient labeling scheme, then so do $\her(\cF^\square)$ and
$\mon(\cF^\square)$.
\end{corollary}

One of our main motivations was to find explicit counterexamples to the IGC; a consequence of the
above corollary is that, counterexamples to the IGC cannot be obtained by taking the monotone
closure of Cartesian products of some hereditary class $\cF$, unless $\cF$ itself is already a
counterexample.  This leaves open the problem of finding an explicit counterexample to the IGC,
which requires finding the first lower-bound technique for adjacency labelling schemes.

Finally, we note that the encoders and decoders for the labeling schemes in \cref{thm:main}
are efficient procedures, and some of the steps in the encoders are randomized.
We elaborate on this more in \cref{section:complexity}.

\section{Adjacency Labeling Scheme}

\paragraph*{Notation.}
For two binary strings $x,y$, we write $x \oplus y$ for the bitwise XOR.  For two graphs $G$ and
$H$, we will write $G \subset H$ if $G$ is a subgraph of $H$, and $G \subset_I H$ if $G$ is an
\emph{induced} subgraph of $H$. We will write $V(G)$ and $E(G)$ as the vertex and edge set of a
graph $G$, respectively. All graphs in this paper are simple and undirected. 
The \emph{degeneracy} of a graph $G$ is the minimum integer $\dgn$
such that all subgraphs of $G$ have a vertex of degree at most $\dgn$.

\paragraph*{Strategy.}
Suppose $G \subset G_1 \square \dotsm \square G_d$ is a subgraph of a Cartesian product. Then $V(G)
\subseteq V(G_1) \times \dotsm \times V(G_d)$. Let $H \subset_I G_1 \square \dotsm \square G_d$ be
the subgraph induced by $V(G)$, so that $E(G) \subseteq E(H)$. One may think of $G$ as being
obtained from the induced subgraph $H$ by deleting some edges. Then two vertices $x,y \in V(G)$ are
adjacent if and only if:
\begin{enumerate}[itemsep=0pt]
\item There exists exactly one coordinate $i \in [d]$ where $x_i \neq y_i$;
\item On this coordinate, $x_iy_i \in E(G_i)$; and,
\item The edge $xy \in E(H)$ has not been deleted in $E(G)$.
\end{enumerate}
We construct the labels for vertices in $G$ in three phases, which check these conditions
in sequence.

\subsection{Phase 1: Exactly One Difference}

We give two proofs for Phase 1. The first is a reduction to the $k$-Hamming Distance communication
protocol. The second proof is direct and self-contained; it is an extension of the proof of the
labeling scheme for induced subgraphs of hypercubes, in the unpublished note \cite{Har22} (adapted
from \cite{Har20,HWZ21}). In both cases the labels are obtained by the probabilistic method, and are
efficiently computable by a randomized algorithm. 

For any alphabet $\Sigma$ and any two strings $x,y \in \Sigma^d$ where $d \in \bN$, write
$\dist(x,y)$ for the Hamming distance between $x$ and $y$, \ie $\dist(x,y) = |\{i \in [d] : x_i \neq
y_i \}|$.  

For the first proof, we require a result in communication complexity (which we translate into our
terminology). A version with two-sided error appears in \cite{Yao03}, the one-sided error version
below is implicit in \cite{HWZ21} (and may appear elsewhere in the literature, which we did not
find).

\begin{theorem}[\cite{Yao03, HWZ21}]
\label{prop:yao}
There exists a constant $c > 0$ satisfying the following.
For any $k \in \bN$, there exists a function $D : \zo^* \times \zo^* \to \zo$
such that, for any $d \in \bN$ and set $S \subseteq \zo^d$ of size $|S|=n$, there exists a
probability distribution $L$ over functions
$\ell : S \to \zo^{ck^2}$, where for all $x,y \in S$,
\begin{enumerate}
\item If $\dist(x,y) \leq k$ then $\Pru{\ell \sim L}{D(\ell(x), \ell(y)) = 1} = 1$; and,
\item If $\dist(x,y) > k$ then $\Pru{\ell \sim L}{D(\ell(x), \ell(y)) = 0} \geq 2/3$.
\end{enumerate}
\end{theorem}

\noindent
We transform these randomized labels into deterministic labels using standard arguments:
\newcommand{\enc}{\mathsf{enc}}
\begin{proposition}
\label{prop:alt-hamming-distance}
There exists a constant $c > 0$ satisfying the following.
For any $k \in \bN$, there exists a function $D : \zo^* \times \zo^* \to \zo$
such that, for any $d \in \bN$ and set $S \subseteq \zo^d$ of size $|S|=n$, there exists a function
$\ell : S \to \zo^{ck^2\log n}$ where for all $x,y \in S$, $D(\ell(x), \ell(y)) = 1$ if and only if
$\dist(x,y) \leq k$.
\end{proposition}
\begin{proof}
Let $D' : \zo^* \times \zo^* \to \zo$, $c > 0$, and $L$ be the function, the constant, and the
probability distribution given for $S$ by \cref{prop:yao}. 
Let $q = \lceil 2 \log_3 n \rceil$, and let $L'$ be the distribution over functions defined by
choosing $\ell_1, \dotsc, \ell_q \sim L$ independently at random, and setting $\ell(x) = (\ell_1(x),
\ell_2(x), \dotsc, \ell_q(x))$ for each $x \in S$.  Define $D : \zo^* \times \zo^* \to \zo$ such
that
\[
  D(\ell(x), \ell(y)) = \bigwedge_{i=1}^q D'(\ell_i(x), \ell_i(y)) \,.
\]
Observe that, if $x,y \in S$ have $\dist(x,y) \leq k$ then $\Pr{ D(\ell(x), \ell(y)) = 1 } = 1$
since for each $i \in [q]$ we have $\Pr{ D'(\ell_i(x), \ell_i(y)) = 1 } = 1$.
On the other hand, if $x,y \in S$ have $\dist(x,y) > k$, then
\[
  \Pr{ D(\ell(x), \ell(y)) = 1} < (1/3)^q \leq 1/n^2 \,.
\]
By the union bound, the probability that there exist $x,y \in S$ such that $D(\ell(x), \ell(y))$
takes the incorrect value is strictly less than 1. Therefore there exists a fixed function $\ell : S
\to \zo^{c k^2 q}$ satisfying the required conditions, where $c k^2 q = C k^2 \log n$ for an
appropriate constant $C$.
\end{proof}

\noindent
We reduce the problem for alphabets $\Sigma$ to the 2-Hamming Distance labeling problem above.

\begin{lemma}
\label{prop:hamming-distance}
There exists a function $D : \zo^* \times \zo^* \to \zo$ and a constant $c>0$ such that, for any
countable alphabet $\Sigma$, any $d \in \bN$, and any set $S \subseteq \Sigma^d$ of size $|S|=n$,
there exists a function $\ell : S \to \zo^k$ for $k \le c\log n$, where $D(\ell(x), \ell(y)) = 1$ if
and only if $\dist(x,y) = 1$.
\end{lemma}
\begin{proof}
Since $\lceil \log n \rceil$ bits can be added to any $\ell(x)$ to ensure that $\ell(x)$ is unique,
it suffices to construct functions $D, \ell$ where $D(\ell(x), \ell(y)) = 1$ if and only if
$\dist(x,y) \leq 1$, instead of $\dist(x,y) = 1$ exactly.

Since $S$ has $n$ elements, we may assume that $\Sigma$ has a finite number $N$ of elements,
since we may reduce to the set of elements which appear in the strings $S$. We may then identify
$\Sigma$ with $[N]$ and define an encoding $\mathsf{enc} : [N] \to \zo^N$ where for any $\sigma \in
[N]$, $\mathsf{enc}(\sigma)$ is the string that takes value 1 on coordinate $\sigma$, and all other
coordinates take value 0.

Abusing notation, for any $x \in \Sigma^d$, we may now define the concatenated encoding $\enc(x) =
\enc(x_1) \circ \enc(x_2) \circ \dotsm \circ \enc(x_d)$, where $\circ$ denotes concatenation.  It is
easy to verify that for any $x,y \in \Sigma^d$, $\dist(\enc(x), \enc(y)) = 2 \cdot \dist(x,y)$. We
may therefore apply \cref{prop:alt-hamming-distance} with $k=2$ on the set $S' = \{ \enc(x) : x \in
S \}$ to obtain a function $D : \zo^* \times \zo^* \to \zo$, a constant $C > 0$, and a function
$\ell' : S' \to \zo^{C \log n}$ such that for all $x,y \in S$,
\[
  D(\ell'(\enc(x)), \ell'(\enc(y))) = 1 \iff \dist(\enc(x), \enc(y)) \leq 2 \iff \dist(x,y) \leq 1
\,.
\]
We may then conclude the proof by setting $\ell(x) = \ell'(\enc(x))$ for each $x \in S$.
\end{proof}

\noindent
Below, we give an alternative, direct proof that does not reduce to $k$-Hamming Distance.
\begin{proposition}
\label{prop:hamming-distance-binary}
For any set $S \subseteq \zo^d$, there exists a random function $\ell : S \to \zo^4$
such that, for all $x,y \in S$,
\begin{enumerate}[itemsep=0pt]
\item[(1)] If $\dist(x,y) \leq 1$, then $\Pru{\ell}{\dist(\ell(x),\ell(y)) \leq 1 } = 1$, and
\item[(2)] If $\dist(x,y) > 1$, then $\Pru{\ell}{\dist(\ell(x),\ell(y)) \leq 1 } \leq 3/4$.
\end{enumerate}
\end{proposition}
\begin{proof}
Choose a uniformly random map $p : [d] \to [4]$ and partition $[d]$ into four sets $P_j =
p^{-1}(j)$. For each $i \in [4]$, define $ \ell(x)_i := \bigoplus_{j \in P_i} x_j $.

Let $x,y \in S$ and write $w = \ell(x) \oplus \ell(y)$. Note that $\dist(\ell(x),\ell(y)) = |w|$,
which is the number of 1s in $w$. If $\dist(x,y) = 0$ then $\dist(\ell(x),\ell(y)) =0\leq 1$. Now suppose $\dist(x,y) = 1$. For any choice of $p : [d] \to [4]$, one
of the sets $P_i$ contains the differing coordinate and will have $w_i = 1$, while the other three
sets $P_j$ will have $w_j = 0$, so
$\Pru{\ell}{ \dist(\ell(x),\ell(y)) \leq 1 } = 1$.

Now suppose $\dist(x,y) = t \geq 2$. We will show that $|w|\leq 1$ with probability at most $3/4$.
Note that $w$ is obtained by the random process where $\vec 0 = w^{(0)}, w=w^{(t)}$, and $w^{(i)}$
is obtained from $w^{(i-1)}$ by flipping a uniformly random coordinate.

Observe that, for $i \geq 1$, $\Pr{ w^{(i)} = \vec 0 } \leq 1/4$. This is because $w^{(i)}=\vec 0$
can occur only if $|w^{(i-1)}|=1$, so the probability of flipping the 1-valued coordinate is $1/4$.
If $|w^{(i-1)}| \geq 1$ then $\Pruc{}{ |w^{(i)}| \leq 1 }{ |w^{(i-1)}| \geq 1} \leq 1/2$ since
either $|w^{(i-1)}|=1$ and then $|w^{(i)}| = 0 \leq 1$ with probability $1/4$, or $|w^{(i-1)}| \geq
2$ and $|w^{(i)}| = 1$ with probability at most $1/2$. Then, for $t \geq 2$,
\begin{align*}
  \Pr{ |w^{(t)}| \leq 1 }
  &= \Pr{ w^{(t-1)} = \vec 0 } + \Pr{ |w^{(t-1)}| \geq 1 } \cdot \Pruc{}{ |w^{(t)}|=1 }{
     |w^{(t-1)}| \geq 1 } 
  \leq \frac{1}{4} + \frac{1}{2} = \frac{3}{4} \,. \qedhere
\end{align*}
\end{proof}

\begin{proposition}
\label{prop:hamming-distance-sigma}
There exists a function $D : \zo^{4} \times \zo^{4} \to \zo$ such that, for any countable alphabet,
$\Sigma$, any $d \in \bN$, and any $S \subseteq \Sigma^d$ of size $n=|S|$, there exists a random
function $\ell : S \to \zo^{4}$ such that, for all $x,y \in S$,
\begin{enumerate}[itemsep=0pt]
\item[(1)] If $\dist(x,y) \leq 1$, then $\Pru{\ell}{D(\ell(x), \ell(y)) = 1 } = 1$, and
\item[(2)] If $\dist(x,y) > 1$, then $\Pru{\ell}{D(\ell(x), \ell(y)) = 1 } \leq 15/16$.
\end{enumerate}
\end{proposition}
\begin{proof}
For each $\sigma \in \Sigma$ and $i \in [d]$, generate an independently and uniformly random
bit $q_i(\sigma) \sim \zo$.  Then for each $x \in S$ define $p(x) = (q_1(x_1),
\dotsc, q_d(x_d)) \in \zo^d$ and $S' = \{ p(x) : x \in S \}$, and let $\ell'$ be the random
function $S' \to \zo^4$ guaranteed to exist by \cref{prop:hamming-distance-binary}. We define the
random function $\ell : S \to \zo^{4}$ as $\ell(x) = \ell'(p(x))$.
We define $D(\ell(x), \ell(y)) = 1$ if and only if $\dist(\ell'(p(x)), \ell'(p(y))) \leq 1$. 

Let $x,y \in S$. Assume first that $\dist(x,y) \leq 1$. By construction, we have $\dist(p(x),p(y)) \leq 1$. Thus, by \cref{prop:hamming-distance-binary} (1),
\[
\Pr{ D(\ell(x), \ell(y)) = 1} = \Pr{ \dist(\ell'(p(x)), \ell'(p(y))) \leq 1 } = 1 \,.
\]

Suppose now that $\dist(x,y) \geq 2$. Then, there are distinct $i,i' \in [d]$ such that $x_i \neq y_i$ and $x_{i'} \neq y_{i'}$, and therefore,
\[
  \Pr{ \dist(p(x), p(y)) \geq 2 } \geq \Pr{ q_i(x_i) \neq q_i(y_i) 
    \wedge q_{i'}(x_{i'}) \neq q_{i'}(y_{i'}) } = 1/4 \,.
\]
Consequently, by the law of total probability and \cref{prop:hamming-distance-binary} (1) and (2), we have
\begin{align*}
  \Pr{ D(\ell(x), \ell(y)) = 1}
  &= \Pr{ \dist(\ell'(p(x)), \ell'(p(y))) \leq 1 } \\
  &= \Pr{ \dist(\ell'(p(x)), \ell'(p(y))) \leq 1 ~|~ \dist(p(x), p(y)) \leq 1} \cdot \Pr{\dist(p(x), p(y)) \leq 1} \\
  & + \Pr{ \dist(\ell'(p(x)), \ell'(p(y))) \leq 1 ~|~ \dist(p(x), p(y)) \geq 2} \cdot (1-\Pr{\dist(p(x), p(y))  \leq 1}) \\
  &\leq \Pr{ \dist(p(x), p(y)) \leq 1} + 3/4 \cdot (1-\Pr{ \dist(p(x), p(y)) \leq 1})\\
  &=   \Pr{ \dist(p(x), p(y)) \leq 1} \cdot (1-3/4) + 3/4 \\
  &\leq 3/4 \cdot (1-3/4) + 3/4 =  15/16 \,. \qedhere
\end{align*}

\end{proof}

\noindent
The alternative proof of \cref{prop:hamming-distance} now concludes by using
\cref{prop:hamming-distance-sigma} with a nearly identical derandomization argument as in
\cref{prop:alt-hamming-distance}. We note that given our explicit descriptions of
the random functions $\ell$ in
\cref{{prop:hamming-distance-binary},prop:hamming-distance-sigma}, the
derandomization argument of \cref{prop:alt-hamming-distance} can be
made constructive and efficient using the method of conditional
expectations (so that the labels can be constructed deterministically
in time polynomial in $n$).

\subsection{Phase 2: Induced Subgraphs}

After the first phase, we are guaranteed that there is a unique coordinate $i \in [d]$ where $x_i
\neq y_i$. In the second phase we wish to determine whether $x_iy_i \in E(G_i)$. It is convenient to
have labeling schemes for the factors $G_1, \dotsc, G_d$ where we can XOR the labels together while
retaining the ability to compute adjacency.  Define an \emph{XOR-labeling scheme} the same as an
adjacency labeling scheme, with the restriction that for each $s \in \bN$ there is some function
$g_s : \zo^s \to \zo$ such that on any two labels $\ell(x), \ell(y)$ of size $s$, the decoder
outputs $D(\ell(x), \ell(y)) = g_s(\ell(x) \oplus \ell(y))$. We show that any labeling scheme can be
transformed into an XOR-labeling scheme with at most a constant-factor loss:

\begin{lemma}
\label{lemma:xor}
Let $\cF$ be any class of graphs with an adjacency labeling scheme of size $s(n)$. Then
$\cF$ admits an XOR-labeling scheme of size at most $4s(n)$.
\end{lemma}
\begin{proof}
Let $D : \zo^* \times \zo^* \to \zo$ be the decoder of the adjacency labeling scheme for $\cF$, fix
any $n \in \bN$, and write $s = s(n)$. Without loss of generality, we assume that $D$ is symmetric, i.e., $D(a,b) = D(b,a)$ for any $a,b \in \zo^s$. Let $\phi : \zo^s \to \zo^{4s}$ be uniformly randomly chosen, so that for every $z
\in \zo^s$, $\phi(z) \sim \zo^{4s}$ is a uniform and independently random variable. For any two
distinct pairs $\{z_1,z_2\}, \{z_1',z_2'\} \in {\zo^s \choose 2}$ where $z_1 \neq z_2$, $z_1' \neq
z_2'$, and $\{z_1, z_2\} \neq \{z_1', z_2'\}$, the probability that $\phi(z_1) \oplus \phi(z_2) =
\phi(z_1') \oplus \phi(z_2')$ is at most $2^{-4s}$, since at least one of the variables $\phi(z_1),
\phi(z_2), \phi(z_1'), \phi(z_2')$ is independent of the other ones.  Therefore, by the union bound,
\[
  \Pr{ \exists \text{ distinct } \{z_1,z_2\}, \{z_1', z_2'\} :
    \phi(z_1) \oplus \phi(z_2) = \phi(z_1') \oplus \phi(z_2')}
  \leq {2^s \choose 2}^2 2^{-4s} \leq \frac{1}{4} \,.
\]
Then there is $\phi : \zo^s \to \zo^{4s}$ such that each distinct pair $\{z_1,z_2\} \in { \zo^s
\choose 2 }$ is assigned a distinct unique value $\phi(z_1) \oplus \phi(z_2)$. So the function
$\Phi(\{z_1, z_2\}) = \phi(z_1) \oplus \phi(z_2)$ is a one-to-one map ${ \zo^s \choose 2} \to
\zo^{4s}$. Then for any graph $G \in \cF$ on $n$ vertices, with labeling $\ell : V(G) \to \zo^s$, we
may assign the new label $\ell'(x) = \phi(\ell(x))$. On labels $\phi(\ell(x)), \phi(\ell(y)) \in
\zo^s$, the decoder for the XOR-labeling scheme simply computes $D\left(\Phi^{-1}(\phi(\ell(x)) \oplus
\phi(\ell(y)))\right) = D(\ell(x), \ell(y))$.
\end{proof}

\cref{lemma:xor} shows the existence of XOR-labelling schemes, but the
proof is non-constructive and in particular it 
does not provide an efficient algorithm to decode the labels. 
We now present an alternative, slightly more complicated, but constructive version of \cref{lemma:xor}. 
It reduces the number of bits in the
XOR-labeling from $4s(n)$ to $2s(n)+2$, and most importantly, it provides an
efficient and deterministic way
to retrieve $\ell(x)$ and $\ell(y)$ from $\ell(x) \oplus \ell(y)$. The
construction is based on the proof of a result of Lindstr\"om \cite[Theorem 2]{Lin69} about Sidon
sets (see also \cite[Proposition 5.1]{BS85}
for a slightly more general result).

\smallskip

We will need a number of
classical facts on binary fields, which we recall now. The Galois
field $\textrm{GF}(2^m)$ can be constructed as follows: its elements
are the polynomials $P(X)=\sum_{i=0}^{m-1} a_i X^i\in \textrm{GF}(2)[X]$ of degree less than
$m$, which are in one-to-one correspondence with their sequences of coefficients
$a(P):=(a_0,\ldots,a_{m-1})\in \textrm{GF}(2)^m$. Adding two elements $P_1$ and $P_2$ in $\textrm{GF}(2^m)$ corresponds to
adding the two polynomials in $\textrm{GF}(2)[X]$, or equivalently
to computing their sequence of coefficients as $a(P_1)\oplus a(P_2)$ (this can be done in time $O(m)$). Multiplying $P_1$ and $P_2$ corresponds to multiplying the
polynomials in $\textrm{GF}(2)[X]$, and then taking the remainder
modulo some fixed irreducible polynomial of degree $m$ in
$\textrm{GF}(2)[X]$ (such a polynomial can be computed
deterministically in time $\widetilde{O}(m^4)$ \cite{Sho90}\footnote{$\widetilde{O}(\cdot)$ hides a polylogarithmic factor.}). The multiplication can be done in time $O(m^2)$
\cite[Chapter 2]{HAC}. Finally, any
quadratic equation in a field has at most two solutions. In
$\textrm{GF}(2^m)$, these
solutions can be computed explicitly in time $O(m^3)$ \cite{Chen82}. We note that this computation does not use the quadratic formula, which fails in
fields of characteristic 2. We also note that in order to apply Chen's
formula \cite{Chen82}  it is convenient to assume that $m$
is odd (in which case we can apply Theorem 1 from \cite{Chen82} rather
than Theorems 2 and 3, which require additional computations).

\begin{lemma}
\label{lemma:xor2}
Let $\cF$ be any class of graphs with an adjacency labeling scheme of size $s(n)$. Then
$\cF$ admits an XOR-labeling scheme of size at most $2s(n)+2$. Moreover,
given $\ell(x) \oplus \ell(y)$, a decoder can retrieve $\ell(x)$ and
$\ell(y)$ deterministically in time $\widetilde{O}(s(n)^4)$.
\end{lemma}
\begin{proof}
  Let $D : \zo^* \times \zo^* \to \zo$ be the decoder of the adjacency labeling scheme for $\cF$, fix
  any $n \in \bN$, and let $s \in\{ s(n),s(n)+1\}$ be an odd integer.
  We assume that the encoder and the decoder agree on an irreducible polynomial of degree $s$ that is used to define
  the field $\textrm{GF}(2^s)$. Such a polynomial can be computed deterministically in time $\widetilde{O}(s^4)$ \cite{Sho90}.
For any $P\in \textrm{GF}(2^s)$, let $\pi(P):=(P,P^3) \in \textrm{GF}(2^s)\times
\textrm{GF}(2^s)$. We claim that for any $P_1\ne P_2\in
\textrm{GF}(2^s)$, $P_1$ and $P_2$ can be uniquely retrieved, in time $O(s^3)$, from the
entrywise sum $\pi(P_1)+\pi(P_2)$ in $\textrm{GF}(2^s) \times
\textrm{GF}(2^s)$ (which corresponds to the XOR of their
sequences of $2s$ coefficients in $ \textrm{GF}(2)$). This follows
from the fact that if $\pi(P_1)+\pi(P_2)=(A,B) \in \textrm{GF}(2^s)\times
\textrm{GF}(2^s)$,  $P_1$ and $P_2$ 
satisfy the equations $P_1+P_2=A\ne 0$ and $P_1^3+P_2^3=B$ in
$\textrm{GF}(2^s)$. By substituting $P_2=A+P_1$ in the second
equality (and recalling that all computations are done in a field of
characteristic 2), we obtain that $P_1$ and $P_2$ are the two solutions of
the quadratic equation $AP^2+A^2P+(A^3+B)=0$, which can be computed in
time $O(s^3)$ \cite[Theorem 1]{Chen82}.

Now, for any graph $G \in \cF$ on $n$ vertices, with labeling $\ell : V(G) \to \zo^{s(n)}$, 
we replace $\ell(x)$ by the label $\ell'(x)\in \{0,1\}^{2s}$
defined as follows. If $s=s(n)+1$ we first add a 0 at the end of
$\ell(x)$, making it an element of $\{0,1\}^{s}$, or equivalently $\textrm{GF}(2)^s$. Let $P_x\in
\textrm{GF}(2^s)$ be the polynomial whose coefficients are given by
$\ell(x)$. Then we
simply define $\ell'(x)$ as the sequence of $2s$ coefficients of
$\pi(P_x)$. Note that for any two distinct vertices $x,y$,
$\ell'(x)\oplus \ell'(y)$ is equal to the sum $\pi(P_x)+\pi(P_y)$ in
$\textrm{GF}(2^s) \times
\textrm{GF}(2^s)$, and $\{ P_x, P_y \}$ can be uniquely retrieved from
this sum. It follows that $\{ \ell(x), \ell(y) \}$ can be uniquely
retrieved from this sum, and, since the decoder is symmetric,
$D(\ell(x),\ell(y))$ can be computed given only the sum $\ell'(x)\oplus \ell'(y)$, as desired.
\end{proof}

We can now prove the first part of \cref{thm:main}. 
\begin{lemma}
\label{lemma:induced}
Let $\cF$ be a hereditary class of graphs that admits an adjacency labeling scheme of size $s(n)$.
Then $\her(\cF^\square)$ admits an adjacency labeling scheme of size $2s(n) + O(\log n)$.
\end{lemma}
\begin{proof}
By \cref{lemma:xor2}, there is an XOR-labeling scheme for $\cF$ with labels of size $2s(n)+2$. Let $D :
\zo^* \times \zo^* \to \zo$ be the decoder for this scheme, with $D(a,b) = g(a \oplus b)$ for some
function $g$. Design the labels for $\her(\cF^\square)$ as follows. Consider a graph $G \in
\her(\cF^\square)$, so that $G \subset_I G_1 \square G_2 \square \dotsm \square G_d$ for
some $d \in \bN$ and $G_i \in \cF$ for each $i \in [d]$. Since $\cF$ is hereditary, we may assume
that each $G_i$ has at most $n$ vertices; otherwise we could simply replace it with the subgraph of
$G_i$ induced by the vertices $\{ x_i : x \in V(G) \}$.  For each $x = (x_1, \dotsc, x_d) \in V(G)$,
construct the label as follows:
\begin{enumerate}
\item Treating the vertices in each $G_i$ as characters of the alphabet $[n]$, use $O(\log n)$ bits
to assign the label given to $x = (x_1, \dotsc, x_d) \in [n]^d$ by \cref{prop:hamming-distance}.
\item Using $2s(n)+2$ bits, append the vector $\bigoplus_{i \in [d]} \ell_i(x_i)$, where $\ell_i(x_i)$
is the label of $x_i \in V(G_i)$ in graph $G_i$, according to the XOR-labeling scheme for $\cF$.
\end{enumerate}
The decoder operates as follows. Given the labels for $x,y \in V(G)$:
\begin{enumerate}
\item If $x$ and $y$ differ on exactly one coordinate, as determined by the first part of the label,
continue to the next step. Otherwise output ``not adjacent''\!\!.
\item Now guaranteed that there is a unique $i \in [d]$ such that $x_i \neq y_i$, output
``adjacent'' if and only if the following is 1:
\begin{align*}
  D\left( \bigoplus_{j \in [d]} \ell_j(x_j)\,, \bigoplus_{j \in [d]} \ell_j(y_j) \right)
  &= g\left( \bigoplus_{j \in [d]} \ell_j(x_j) \oplus \bigoplus_{j \in [d]} \ell_j(y_j) \right) \\
  &= g\left( \ell_i(x_i) \oplus \ell_i(y_i) \oplus
      \bigoplus_{j \neq i} \ell_j(x_j) \oplus \ell_j(y_j) \right)
  = g( \ell_i(x_i) \oplus \ell_i(y_i) ) \,,
\end{align*}
where the final equality holds because $x_j = y_j$ for all $j \neq i$, so $\ell_j(x_j) =
\ell_j(y_j)$. Then the output value is 1 if and only $x_iy_i$ is an edge of $G_i$; equivalently,
$xy$ is an edge of $G$.
\end{enumerate}
This concludes the proof.
\end{proof}


The XOR-labeling trick can also be used to simplify the proof of \cite{HWZ21} for adjacency
\emph{sketches} of Cartesian products. That proof is similar to the one above, except it uses a
two-level hashing scheme and some other tricks to avoid destroying the labels of $x_i$ and $y_i$
with the XOR (with sufficiently large probability of success). This two-level hashing approach does
not succeed in our current setting, and we avoid it with XOR-labeling.

\subsection{Phase 3: Subgraphs}

Finally, we must check whether the edge $xy \in E(H)$ in the \emph{induced} subgraph $H \subset_I
G_1 \square \dotsm \square G_d$ has been deleted in $E(G)$. There is a minimal and perfect tool for
this task:

\begin{theorem}[Minimal Perfect Hashing]
\label{thm:perfect-hashing}
For every $m, k \in \bN$, and any $S \subseteq [m]$ of size $k$, there exists a function $h : [m] \to [k]$ where the image of $S$ under $h$ is $[k]$ and for every distinct $i,j \in S$ we have $h(i) \neq h(j)$. The function $h$ can be
stored in $k \ln e + \log\log m + o(k + \log\log m)$ bits of space and it can be computed by a
randomized algorithm in expected time $O(k + \log\log m)$.
\end{theorem}

Minimal perfect hashing has been well-studied. A proof of the space bound appears in \cite{Meh84}
and significant effort has been applied to improving the construction and evaluation time. We take
the above statement from \cite{HT01}.  We note that the randomized
computation of $h$ can be replaced by a (slightly less efficient)
deterministic computation at the cost of a multiplicative factor $\log k$ in the
storage space of $h$ \cite{AN96}.

\smallskip

We now conclude the proof of \cref{thm:main} by applying the
next lemma to the class $\cG = \her(\cF^\square)$, using the labeling scheme for $\her(\cF^\square)$
obtained in \cref{lemma:induced} (note that $\mon(\her(\cF^\square)) = \mon(\cF^\square)$).

\begin{lemma}
\label{lemma:hashing}
Let $\cG$ be any hereditary graph class which admits an adjacency labeling scheme of size $s(n)$.
Then $\mon(\cG)$ admits an adjacency labeling scheme where each $G \in \mon(\cG)$ on $n$ vertices
has labels of size $s(n) + O(\dgn(n) + \log n)$, where $\dgn(n)$ is the maximum degeneracy of any $n$-vertex graph in $\cG$.
\end{lemma}
\begin{proof}
Let $G \in \mon(\cG)$ have $n$ vertices, so that it is a subgraph of $H \in \cG$ on $n$ vertices.
The labeling scheme is as follows.
\begin{enumerate}
\item Fix a total order $\prec$ on $V(H)$ such that each vertex $x$ has at most $\dgn=\dgn(n)$ neighbors
$y$ in $H$ with $x \prec y$; this exists by the definition of degeneracy. We will identify each vertex $x$ with
its position in the order.
\item For each vertex $x$, assign the label as follows:
\begin{enumerate}
\item Use $s(n)$ bits for the adjacency label of $x$ in $H$.
\item Use $\log n$ bits to indicate the position of $x$ in the order.
\item Let $N^+(x)$ be the set of neighbors $y$ of $x$ in $H$ with $x \prec y$,
and denote $d_x = |N^+(x)|$.
Construct a perfect hash function $h_x : N^+(x) \to [d_x]$ and store it using $O(d_x + \log\log n)=O(\dgn + \log\log n)$ bits.
\item Use $d_x \leq \dgn$ bits to write the function $\mathsf{edge}_x : [d_x] \to \zo$ which takes value 1 on $i \in [d_x]$ if and only if $xy$ is an edge of $G$, where $y$ is the unique vertex in $N^+(x)$
satisfying $h_x(y) = i$.
\end{enumerate}
\end{enumerate}
Given the labels for $x$ and $y$, the decoder performs the following:
\begin{enumerate}
\item If $xy$ are not adjacent in $H$, output ``not adjacent''\!\!.
\item Otherwise $xy$ are adjacent. If $x \prec y$, we are guaranteed that $y$ is in the domain of
$h_x$, so output ``adjacent'' if and only if $\mathsf{edge}_x(h_x(y))=1$. If $y \prec x$, output
``adjacent'' if and only if $\mathsf{edge}_y(h_y(x))=1$.
\end{enumerate}
This concludes the proof.
\end{proof}

\section{Optimality}
\label{section:optimality}

We now prove the optimality of our labeling schemes, and \cref{cor:efficient}. We require:

\begin{proposition}
\label{prop:monotone-degeneracy}
For any hereditary class $\cF$, let $\dgn(n)$ be the maximum degeneracy of an $n$-vertex graph in $\her(\cF^\square)$. Then, for every $n \in \bN$, the class $\her(\cF^\square)$ contains a graph $H$ on $n$ vertices with at least
$n \cdot \dgn(n) / 4$ edges, so $\mon(\cF^\square)$ contains all $2^{n\cdot \dgn(n)/4}$ spanning subgraphs
of $H$.
\end{proposition}
\begin{proof}
Fix an arbitrary $n \in \bN$ and let $G$ be an $n$-vertex graph in $\her(\cF^\square)$ of degeneracy $\dgn = \dgn(n)$.
By definition, $G$ contains an induced subgraph $G' \subset_I G$ with minimum
degree $\dgn$ and $n_1 \leq n$ vertices. If $n_1 \geq n/2$ then $G$ itself has at least $\dgn n_1/2 \geq \dgn n/4$ edges, and we are done. Now assume $n_1 < n/2$. Since $G \in \her(\cF^\square)$, $G \subset_I
H_1 \square \dotsm \square H_t$ for some $t \in \bN$ and $H_i \in \cF$. So for any $d \in \bN$, the
graph $(G')^d \subset_I (H_1 \square \dotsm \square H_t)^d$ belongs to $\her(\cF^\square)$. Consider
the graph $H \subset_I (G')^d$ defined as follows. Choose any $w \in V(G')$, and for each $i \in [d]$ let
\[
V_i = \{ (v_1, v_2, \dotsc, v_d) : v_i \in V(G') \text{ and } \forall j \neq i, v_j = w \}\,,
\]
and let $H$ be the graph induced by vertices $V_1 \cup \dotsm \cup V_d$. Then $H$ has $dn_1$
vertices, each of degree at least $\dgn$, since each $v \in V_i$ is adjacent to $\dgn$ other vertices in $V_i$. Set $d = \lceil n/n_1 \rceil$, so that $H$ has at least $n$ vertices, and let $m = dn_1-n$, which satisfies $m < n_1$. Remove any $m$ vertices of $V_1$. The remaining graph $H'$ has $n$
vertices, and at least $(d-1)n_1 \geq n-n_1 > n/2$ vertices of degree at least $\dgn$. Then $H'$ has at least $\dgn n/4$ edges.
\end{proof}

The next proposition shows that \cref{thm:main} is optimal up to constant factors. It is
straightforward to check that this proposition implies \cref{cor:efficient}.

\begin{proposition}
Let $\cF$ be a hereditary class whose optimal adjacency labeling scheme has size $s(n)$ and which
contains a graph with at least one edge. Then any adjacency labeling scheme for $\her(\cF^\square)$
has size at least $\Omega(s(n) + \log n)$, and any adjacency labeling scheme for $\mon(\cF^\square)$
has size at least $\Omega(s(n) + \dgn(n) + \log n)$, where $\dgn(n)$ is the maximum degeneracy of any
$n$-vertex graph in $\mon(\cF^\square)$.
\end{proposition}
\begin{proof}
Since $\cF \subseteq \her(\cF^\square)$ and $\cF \subseteq \mon(\cF^\square)$, we have a lower bound
of $s(n)$ for the labeling schemes for both of these classes. Since $\cF$ contains a graph $G$ with
at least one edge, the Cartesian products contain the class of hypercubes: $\her(\{K_2\}^\square)
\subseteq \her(\cF^\square) \subseteq \mon(\cF^\square)$. A labeling scheme for
$\her(\{K_2\}^\square)$ must have size $\Omega(\log n)$ (which can be seen since each vertex of
$K_2^d$ has a unique neighborhood and thus requires a unique label). This establishes the lower
bound for $\her(\cF^\square)$, since the labels must have size $\max\{ s(n), \Omega(\log n)\} =
\Omega(s(n) + \log n)$. Finally, by \cref{prop:monotone-degeneracy}, the number of $n$-vertex graphs
in $\mon(\cF^\square)$ is at least $2^{\Omega(n \dgn(n))}$, so there is a lower bound on the label size
of $\Omega(\dgn(n))$, which implies a lower bound of $\max\{ s(n), \Omega(\log n), \Omega(\dgn(n))\} =
\Omega(s(n) + \dgn(n) + \log n)$ for $\mon(\cF^\square)$.
\end{proof}

\section{Time Complexity}
\label{section:complexity}
Let $\cF$ be a hereditary graph class with an adjacency labeling scheme whose encoder and decoder
are available as black box algorithms.  Then our decoders for $\her(\cF^\square)$ and
$\mon(\cF^\square)$ are deterministic and work in polynomial time in the size of the labels.  The
encoders we have described are randomized and produce correct labels
in expected time polynomial in the number of vertices, but they can be made
deterministic with a small loss on the complexity (and the size of the
labels, see the remark after \cref{thm:perfect-hashing}). 
The encoding algorithm in Phase 3 requires to be given the input graph $G \in
\mon(\cG)$ (in our application, $\cG = \cF^\square$) together with a graph $H \in \cG$ on the same
vertex set that contains $G$ as a subgraph.  How exactly the graph $H$ might be determined is left
unspecified and depends on the original class $\cG$.

\vspace{2em}
\begin{center}
\textbf{\large Acknowledgments}
\end{center}
We are very grateful to Sebastian Wild, who prevented us trying to reinvent perfect hashing.
We also thank the anonymous reviewers for their comments, in particular for pushing us to get
time-efficient encoding and decoding algorithms.


\bibliographystyle{alpha}
\bibliography{references.bib}
\end{document}